\documentclass[12pt,twoside,reqno]{amsart}
\usepackage{pifont}
\usepackage{amsmath}
\usepackage{amsthm}
\usepackage{amsfonts}
\usepackage{amssymb}
\usepackage{latexsym}
\usepackage{amstext}
\usepackage{array}
\usepackage{graphicx}

\date{}
\allowdisplaybreaks[4] \footskip=15pt
\renewcommand{\uppercasenonmath}[1]{}

\newtheorem{Theorem}{Theorem}[section]
\newtheorem{Corollary}{Corollary }[section]
\newtheorem{Lemma}{Lemma}[section]
\newtheorem{Remark}{Remark}[section]

\newtheorem{Proposition}{Proposition}[section]

\theoremstyle{definition}
\newtheorem{Definition}{Definition}[section]

\usepackage{amscd,amsthm,amsfonts,amssymb,amsmath}
\usepackage[colorlinks=true]{hyperref}
\usepackage{fullpage}
\usepackage[all]{xy}

    \numberwithin{equation}{section}

\input txdtools
\begin{document}

\begin{center}

{\large  \bf Regular complete permutation polynomials over quadratic extension fields}

\vskip 0.8cm
{\small Wei Lu$^1$ $\cdot$ Xia Wu$^1$ \footnote{Supported by NSFC (Nos. 11971102, 11801070, 11771007), the Fundamental Research Funds for the Central Universities.

  }} $\cdot$ Yufei Wang$^1$   $\cdot$ Xiwang Cao$^2$ \\

{\small $^1$School of Mathematics, Southeast University, Nanjing
210096, China}\\
{\small $^2$Department of Math, Nanjing University of Aeronautics and Astronautics, Nanjing 211100, China}\\
{\small E-mail:
 luwei1010@139.com, wuxiadd1980@163.com, 220211734@seu.edu.cn, xwcao@nuaa.edu.cn}\\
{\small $^*$Corresponding author. (Email: wuxiadd1980@163.com)}
\vskip 0.8cm
\end{center}

{\bf Abstract:} Let $r\geq 3$ be any positive integer  which is  relatively prime to $p$ and $q^2\equiv 1 \pmod r$.  Let $\tau_1, \tau_2$ be any permutation polynomials over $\mathbb{F}_{q^2},$ $\sigma_M$ is an invertible linear map  over  $\mathbb{F}_{q^2}$ and $\sigma=\tau_1\circ\sigma_M\circ\tau_2$. In this paper, we prove that, for suitable $\tau_1, \tau_2$ and $\sigma_M$, the map $\sigma$ could be $r$-regular complete permutation polynomials over quadratic extension fields.


\section{\bf Introduction}

Let $p$ be a prime, $q$ be a power of $p$ and $\mathbb{F}_q$ the finite field with $q$ elements. A polynomial $f(x)\in{\mathbb{F}}_q[x]$ is called a \emph{permutation polynomial} (PP) over ${\mathbb{F}}_q$ if the associated polynomial function $f: c\mapsto f(c)$ from ${\mathbb{F}}_q$ into ${\mathbb{F}}_q$ is a permutation of ${\mathbb{F}}_q$.
A polynomial $f(x)\in{\mathbb{F}}_q[x]$ is called a \emph{complete permutation polynomial} (CPP) over ${\mathbb{F}}_q$ if  both $f(x)$ and $f(x)+x$ are permutations of  ${\mathbb{F}}_q$. These polynomials were introduced by Mann in the construction of orthogonal Latin squares \cite{M1942}. Niederreiter and Robinson later gave a detailed study of CPPs over finite fields \cite{NR1982}.  CPPs have widely applications in the design of nonlinear dynamic substitution device \cite{M1995,M1997},  the Lay-Massey scheme \cite{V1999}, the block cipher SMS4 \cite{DL2008}, the stream cipher Loiss \cite{FFZ2011}, the design of Hash functions \cite{SV1995, V1994}, quasigroups \cite{MM2009,MM2012a,MM2012b}, and  the constructions of some cryptographically strong functions \cite{MP2014,SG2012,ZHC2015}.
The two most important concepts related to a permutation  are the existence of fixed points and the specification of its cycle structure.

 Let $f(x)$ be a PP over ${\mathbb{F}}_q$,  $e$  the identity map over ${\mathbb{F}}_q$ and   $n$  a positive integer. Let $f^{(n)}$ be   the $n$-th composite power of $f$, which is defined inductively by $f^{(n)}:=f\circ f^{(n-1)}=f^{(n-1)}\circ f$ and $f^{(1)}:=f,$ $f^{(0)}:=e$, $f^{(-n)}:=(f^{-1})^{(n)}$.    The  polynomial $f(x)$ is called  an \emph{$n$-cycle permutation} if $f^{(n)}$ equals  the identity map $e$. On the one hand, PPs with long cycles (especially full cycles) can be used to generate key-stream sequences with large periods \cite{F1982, G1967, GG2005}. On the other hand,  PPs with short cycles (especially  involutions $f^{(2)}=e$) can be used   to construct bent functions over finite fields \cite{CM2018,G1962,M2016},  to design codes \cite{G1962} and to  against  some cryptanalytic attacks \cite{CR2015}. In general, it is difficult to  determine the cycle structure of a PP. Very few known PPs whose explicit cycle structures have been obtained (see \cite{A1969, LM1991, RC2004, RMCC2008} for monomials and Dickson polynomials).

 A polynomial $f(x)\in{\mathbb{F}}_q[x]$ is called   an \emph{$r$-regular} PP over ${\mathbb{F}}_q$ if $f(x)$ is a PP over ${\mathbb{F}}_q$ with all the cycles of the same length $r$ (ignoring the fixed points). Regular PPs are very important in applications of turbo-like coding, low-density parity-check codes (LDPC) and block cipher designs \cite{B2003,RC2004,RMCC2008,SSP2012}. In \cite{MP2014}, a recursive construction of CPPs  over finite fields
via subfield functions was proposed to construct CPPs  with no fixed points  over   finite fields  with odd characteristic. The similar technique was used  to construct strong complete mappings  in \cite{M2021} and to construct some  $r$-regular CPPs over ${\mathbb{F}}_q$  with even characteristic for some small positive integers $r$ in \cite{XZZ2022}.  In \cite{LWWC2022}, we generalized the technique used in \cite{MP2014,M2021,XZZ2022} and give a general construction of regular PPs and regular CPPs over extension fields. The maps considered in \cite{LWWC2022} are of the forms $f=\tau\circ\sigma_M\circ\tau^{-1}$ where $\tau$ is a PP over an extension field $\mathbb{F}_{q^d}$ and $\sigma_M$ is an invertible linear map  over  $\mathbb{F}_{q^d}$. In order to get regular CPPs, the difficulty is to make sure that $f+e$ is   a PP over $\mathbb{F}_{q^d}$ when $\tau$ is not additive. In \cite{LWWC2022} we  give several examples of regular CPPs for $r=3,4,5,6,7$ and for arbitrary odd positive integer $r$. For different examples, we  used different methods to prove that $f+e$ is a PP.

The main purpose of this paper is to give new regular CPPs  over  quadratic extension fields $\mathbb{F}_{q^2}$  for arbitrary positive integer $r$. In fact, CPPs  over  quadratic extension fields are enough for practical application. For example, let $p(z)$ be a polynomial from $\mathbb{F}_{q}$ to itself, let $\Omega_p,$ $\Phi_p$ and $\Psi_p$ be three mappings from $\mathbb{F}_{q^2}$ to itself defined by $\Omega_p(x_1, x_2)=(x_2, p(x_2)+x_1),$ $\Phi_p(x_1, x_2)=(x_2, p(x_1)-x_2),$ and $\Psi_p(x_1, x_2)=(p(x_2), p(x_2)+x_1),$ where $x=(x_1, x_2) \in \mathbb{F}_{q^2}.$ Then the maps  $\Omega_p,$ $\Phi_p$ and $\Psi_p$ are CPPs over $\mathbb{F}_{q^2},$  when $p(z)$ is a PP over $\mathbb{F}_{q}$ (see \cite[Lemma 3]{XLZH2018}).  The maps  $\Omega_p,$ $\Phi_p$ and $\Psi_p$ correspond to one-round Feistel structure, L-MISTY structure and R-MISTY structure without round key, respectively. Modern block ciphers base their design on an inner function that is iterated a high number of times. The Feistel and MISTY structures have been used in the design of many block ciphers (in particular the DES \cite{NIST1993} and MISTY \cite{Ma1997}).



This paper is organized as follows. In section \ref{section Preliminaries}, we introduce some basic knowledge about cycle structure of permutation polynomials, cyclotomic polynomials over finite fields and some properties of linear maps in linear algebra. In section \ref{Section 3}, We will give two constructions of $r$-regular  CPPs over quadratic extension fields  for any positive integer $r$ which satisfies $q^2\equiv 1 \pmod r$.
In section \ref{section Concluding remarks}, we conclude this paper.

 \section{\bf Preliminaries}\label{section Preliminaries}

\subsection{Cycle structures of permutation polynomials}

 In this subsection, we prepare and discuss the  cycle structures of permutation polynomials.

   Let $q$ be a prime power and $\mathbb{F}_q$ the finite field with $q$ elements.


 \begin{Definition}\cite[Definition 4]{CMS2016}
 Let $f$ be a PP over ${\mathbb{F}}_q$ and $t$  a positive integer. A \emph{cycle} of $f$ is a subset ${x_1,\dots , x_t}$ of pairwise distinct elements of ${\mathbb{F}}_q$ such that $f(x_i)=x_{i+1}$ for $1\leq i\leq t-1$ and $f(x_t)=x_1$. The cardinality of a cycle is called its \emph{length}.
 \end{Definition}

%
%
%
%
%
%

 \begin{Proposition}\cite[Proposition 2.11]{CWZ2021}\label{proposite 2.2}
Let $f$ and $g$ be  PPs over ${\mathbb{F}}_q$.  Furthermore, $f$ is an $n$-cycle permutation. Then $g\circ f\circ g^{-1}$ is also an $n$-cycle permutation. Moreover, $f$ and $g\circ f\circ g^{-1}$ have the same cycle structures.
 \end{Proposition}

 Now we recall the definitions of  regular PPs and  regular CPPs over ${\mathbb{F}}_q$.

\begin{Definition}\cite[Definition 1]{XZZ2022}
A PP $f$ over ${\mathbb{F}}_q$ is called  \emph{$r$-regular} if   all the cycles of $f$ have the same length $r$ (ignoring the fixed points).
\end{Definition}

\begin{Remark}\label{Remark 2.1}
 It is easy to see that,  if  $r$ is a prime, then  any non-identity  $r$-cycle permutation is $r$-regular.
\end{Remark}

\begin{Definition}\cite[Definition 2]{XZZ2022}
A CPP $f$ over ${\mathbb{F}}_q$ is called \emph{$r$-regular}   if   all the cycles of $f$ have the same length $r$ (ignoring the fixed points).
\end{Definition}

\begin{Remark}
 It is easy to see that,  when the characteristic of  ${\mathbb{F}}_q$  is even,   any CPP $f$ has only one fixed point.
\end{Remark}

\subsection{ Linear algebra over finite fields}

In this subsection, we recall a few concepts and facts from linear algebra \cite[p.60, 64, 525]{R2003}.

Let  $d$ be a positive integer and $V:=\mathbb{F}_q^d$   the column   vector space with dimension $d$ over $\mathbb{F}_q.$ Let  $M_{d\times d}(\mathbb{F}_q)$ be the matrix ring over $\mathbb{F}_q$ and $M$  any $d\times d$ matrix in  $M_{d\times d}(\mathbb{F}_q)$.  Then we can define  a \emph{linear map} $\sigma_M$ from $V$ to $V$ by the usual way
 \begin{equation}\label{linear map}
 \sigma_M(v)=Mv,{\ \rm{ for\  any \ }} v\in V.
 \end{equation}

\begin{Definition}
Let $M$ be any $d\times d$ matrix over $\mathbb{F}_q$ and $\sigma_M$ the linear map associate with $M$. We define the \emph{characteristic polynomial} $P_{\sigma_M}(t):=P_M(t)$ to be the determinant
\begin{center}
$\mathrm{det}$$(tI_d-M)$
\end{center}
where $I_d$ is the unit $d\times d$ matrix. The characteristic polynomial is an element of $\mathbb{F}_q[t]$.

 \end{Definition}




The following result  is important in the proof of  our main theorem.

\begin{Theorem}\cite[p.561]{L2002}{(Cayley-Hamilton Theorem)}
Let $M$ be any $d\times d$ matrix over $\mathbb{F}_q.$ We have
\begin{center}
 $P_{\sigma_M}({\sigma_M})=0.$
\end{center}
\end{Theorem}


If $\sigma$ is a linear map on the vector space $V=\mathbb{F}_q^d$, then a polynomial $h(t)\in \mathbb{F}_q[t]$ is said to \emph{annihilate} $\sigma$ if $h(\sigma)=0$, where 0 is the zero map on $V$. The uniquely determined monic polynomial of least positive degree with this property is called the \emph{minimal polynomial} of $\sigma$. Minimal polynomial  divides any other polynomial in $\mathbb{F}_q[x]$ annihilating $\sigma$. In particular, the minimal polynomial of $\sigma$ divides the characteristic polynomial $P_\sigma$ by Cayley-Hamilton Theorem.

%
%
%


Next, we recall some results about cyclotomic polynomials \cite[p.64]{R2003}.
\begin{Definition}\cite[Definition 2.44]{R2003}
Let $n$ be a positive integer not divisible by $p$, and $\zeta$ a primitive $n$-th root of unity over $\mathbb{F}_q$. Then the polynomial
$$Q_n(x)=\prod\limits_{s=1\atop\gcd (s,n) = 1}^{n} {(x - {\zeta ^s})}$$
is called the \emph{n}-th cyclotomic polynomial over $\mathbb{F}_q$. The degree of $Q_n(x)$ equals $\phi(n)$, where $\phi(n)$  is Euler's function and indicates the number of integers $s$ with $1\leq s\leq n$ that are relatively prime to $n$.
\end{Definition}
The following results are basic.

\begin{Proposition}\cite[Theorem 2.47, Lemma 2.50]{R2003}
Let $n$ be a positive integer not divisible by $p$. Then
\begin{enumerate}
  \item $x^n-1=\prod\limits_{l|n}{Q_l(x)}$;
  \item if $l$ is a divisor of $n$ with $1\leq l<n$, then $Q_n(x)$ divides $(x^n-1)/(x^l-1)$.
\end{enumerate}
\end{Proposition}

\begin{Lemma}\cite[p.65]{R2003}
If $(n,q)=1$, then $Q_n(x)$ factors into $\phi(n)/d$ distinct monic irreducible polynomials in $\mathbb{F}_q[x]$ of the same degree $d$, where $d$ is the least positive integer such that $q^d\equiv 1 \pmod n$.
\end{Lemma}

The following corollary is trivial.

\begin{Corollary}\label{corollay quadratic}
If $q^2\equiv 1 \pmod n$, then $Q_n(x)$ has a quadratic factor in $\mathbb{F}_q[x]$.
\end{Corollary}

 In Corollary \ref{corollay quadratic}, when $q\equiv 1  \pmod  n$, there are ${\phi(n)\choose 2}=\frac{\phi(n)(\phi(n)-1)}{2}$ distinct quadratic polynomials satisfy the condition; while when $q\not\equiv 1\pmod n$ and $q^2\equiv 1 \pmod n$, there are $\frac{\phi(n)}{2}$ distinct quadratic polynomials satisfy the condition.
In fact,  such quadratic factors can be constructed in the following way.

\begin{Proposition}
\begin{enumerate}
  \item When $q\equiv 1  \pmod  n$, let $\mathbb{F}_{q}^*=<\alpha>$ and $\zeta=\alpha^{\frac{q-1}{n}}$. Then  $\zeta$ is a primitive $n$-th root of unity and $\zeta\in \mathbb{F}_{q}$. Let $s,t$ be any positive integers with $1\leq s, t\leq n$ and $\gcd(s,n)=\gcd(t,n)=1.$ Let $h(x)=h_{s,t}(x)=(x-\zeta^s)(x-\zeta^t).$   Then $h(x)\in\mathbb{F}_q[x]$ and $h(x)$ is a quadratic factor of $Q_n(x)$.
        \item When $q\not\equiv 1\pmod n$ and $q^2\equiv 1 \pmod n,$ let $\mathbb{F}_{q^2}^*=<\alpha>$ and $\zeta=\alpha^{\frac{q^2-1}{n}}$. Then  $\zeta$ is a primitive $n$-th root of unity and $\zeta\in \mathbb{F}_{q^2}\backslash \mathbb{F}_{q}$. Let $s$ be any positive integer with $1\leq s\leq n$ and $\gcd(s,n)=1.$ Let $h(x)=h_{s}(x)=(x-\zeta^s)(x-\zeta^{sq}).$   Then $h(x)\in\mathbb{F}_q[x]$ and $h(x)$ is a quadratic factor of $Q_n(x)$.
\end{enumerate}
%
%
\end{Proposition}

\section{\bf Constructions   of regular  CPPs over quadratic extension fields}\label{Section 3}

In this section, let $r\geq 3$ be any positive integer  which is  relatively prime to $p$ and $q^2\equiv 1 \pmod r$. We will give two constructions of $r$-regular  CPPs over quadratic extension fields.

\begin{Theorem}\label{main theorem}
Let $r\geq 3$ be any positive integer, which is relatively prime to $p$. Assume that $q^2\equiv 1 \pmod r$. By Corollary \ref{corollay quadratic}, let $h(t)=t^2+h_1t+h_0\in\mathbb{F}_q[t]$ be a factor of $Q_r(t)$ and $M\in M_{2\times2}(\mathbb{F}_q)$ satisfies that $P_M(t)=h(t)$. Let $\tau_1, \tau_2$ be any PPs over $\mathbb{F}_{q^2}$ and $\sigma=\tau_1\circ\sigma_M\circ\tau_2$. Then
\begin{enumerate}
  \item $\sigma$ is a  PP over $\mathbb{F}_{q^2};$
  \item if $\tau_1\circ\tau_2=e$, then $\sigma$ is an $r$-regular PP over $\mathbb{F}_{q^2};$
  \item if $\tau_1\circ\tau_2=e$ and $\tau_1$ is additive, then $\sigma$ is an $r$-regular CPP over $\mathbb{F}_{q^2}$;
  \item if
           \begin{equation*}
           M=\begin{pmatrix}
            -h_0      & mh_0(h_0-h_1+1)\\ -m^{-1} &    h_0-h_1
            \end{pmatrix}
           \end{equation*}

  where $m\in \mathbb{F}_q^*$, $\tau_1(x_1,x_2)=(a_1(x_1),x_2)$ and $\tau_2(x_1,x_2)=(a_2(x_1),x_2)$, where $a_1, a_2$ are any PPs over $\mathbb{F}_q$, then $\sigma$ is a  CPP over $\mathbb{F}_{q^2}$. Moreover, if $a_1\circ a_2=e,$ then $\sigma$ is an r-regular CPP over $\mathbb{F}_{q^2};$
  \item if
      \begin{equation*}
        M=\begin{pmatrix}
            -h_1+1      & m(h_0-h_1+1)\\  -m^{-1} &    -1
            \end{pmatrix}
      \end{equation*}

  where $m\in \mathbb{F}_q^*$, $\tau_1(x_1,x_2)=(a_1(x_1),x_2)$ and $\tau_2(x_1,x_2)=(a_2(x_1),x_2)$, where $a_1, a_2$ are any PPs over $\mathbb{F}_q$, then $\sigma$ is a  CPP over $\mathbb{F}_{q^2}$. Moreover, if $a_1\circ a_2=e,$ then $\sigma$ is an r-regular CPP over $\mathbb{F}_{q^2}.$
\end{enumerate}
\end{Theorem}
\begin{proof}
\begin{enumerate}
  \item Note that $0\neq h(0)=P_M(0)=\mathrm{det}$$(0I_2-M)=\mathrm{det}$$(-M)=(-1)^2\mathrm{det}(M)$. Hence $M$ is an invertible matrix and $\sigma_M$ is an invertible linear map. Combining that $\tau_1$ and $\tau_2$ are PPs over $\mathbb{F}_{q^2}$, we have $\sigma=\tau_1\circ\sigma_M\circ\tau_2$ is a PP over $\mathbb{F}_{q^2}$.
  \item We only need to prove that $\sigma$ is $r$-regular. First, we prove that $\sigma_M$ is $r$-regular. For any $x\in \mathbb{F}_{q^2}^*$, let $l=l_x$ be its length in  $\sigma_M$. That is to say, $x, {\sigma_M}x,\ldots,\sigma_M^{l-1}x$ are pairwise distinct and ${\sigma_M^l}x=x$. Let
$$W:={\mathbb{F}_q}x+{\mathbb{F}_q}{\sigma_M}x+\cdots+{\mathbb{F}_q}{\sigma_M^{l-1}}x.$$
Then $W$ is an invariant subspace of $\sigma_M$. Moreover, $h(t)$ and $t^l-1$ are two annihilating polynomials of $\sigma_M|_W$.
Let $g(t)$ be the minimal polynomial of $\sigma_M|_W$. Then
$$g(t)\ |\ (h(t),t^l-1).$$
By $h(t)\ |\ Q_r(t)$, we have $g(t)\ |\ (Q_r(t),t^l-1)$. If $l<r$, then by \cite[Lemma 2.50, p.66] {R2003}, $(Q_r(t),t^l-1)=1$ and $g(t)=1$, that is a contradiction. So $l=r$ and $\sigma_M$ is $r$-regular. By Proposition \ref{proposite 2.2}, we have $\sigma=\tau_1\circ\sigma_M\circ\tau_1^{-1}$ is also $r$-regular.
  \item We only need to prove that $\sigma+e$ is a permutation of $\mathbb{F}_{q^2}$. First, we prove that $\sigma_M+e$ is a permutation of $\mathbb{F}_{q^2}$. Since $(r,p)=1$, $r\geq3$ and $h(t)\mid Q_r(t)$, we have   $h(-1)\neq0$. Note that the matrix of $\sigma_M+e$ is $M+I_2$ and its characteristic polynomial is $P_{\sigma_M+e}(t)=\mathrm{det}$$(tI_2-(M+I_2))=\mathrm{det}$$((t-1)I_2-M)=P_M(t-1)=h(t-1)$. Hence $(-1)^2\mathrm{det}(M+I_2)=P_{\sigma_M+e}(0)=h(-1)\neq 0$, $M+I_2$ is an invertible matrix and $\sigma_M+e$ is an invertible linear map. Next, since $\tau_1$ is additive, we have
\begin{equation*}
      \tau_1\circ(\sigma_M+e)\circ\tau_1^{-1}=\tau_1\circ(\sigma_M\circ\tau_1^{-1}+\tau_1^{-1})=\tau_1\circ\sigma_M\circ\tau_1^{-1}+\tau_1\circ\tau_1^{-1}=\sigma+e.
\end{equation*}
       By Proposition \ref{proposite 2.2}, we have $\sigma+e=\tau_1\circ(\sigma_M+e)\circ\tau_1^{-1}$ is a PP over $\mathbb{F}_{q^2}$.
  \item It is easy to check that $P_{M}(t)=h(t),$ by (1), $\sigma$ is a  PP over $\mathbb{F}_{q^2}.$  Next, we  prove that $\sigma+e$ is PP over $\mathbb{F}_{q^2}$. Let
  \begin{equation}
  \begin{aligned}
  \tau_3\begin{pmatrix}
  x_1\\
  x_2\\
  \end{pmatrix}
  =
  \begin{pmatrix}
  a_1(x_1)+a_2^{-1}(x_2)\\
  -m^{-1}h_0^{-1}x_1\\
  \end{pmatrix}
  ,
  \tau_4\begin{pmatrix}
  x_1\\
  x_2\\
  \end{pmatrix}
  =
  \begin{pmatrix}
  mh_0x_2\\
  a_2(x_1)-mh_0x_2\\
  \end{pmatrix}
  \end{aligned}
  \end{equation}
  and
  \begin{equation*}
    M_1=\begin{pmatrix}
           -h_1+1 & -h_0\\
            1 &  1
            \end{pmatrix}.
  \end{equation*}
  It is easy that $\tau_3,\tau_4,\sigma_{M_1}$ are PPs over $\mathbb{F}_{q^2}$. By calculating, we have
  \begin{equation}
  \begin{aligned}
  (\sigma+e)\begin{pmatrix}
  x_1\\
  x_2\\
  \end{pmatrix}
  =
  \begin{pmatrix}
  a_1(-h_0a_2(x_1)+mh_0(h_0-h_1+1)x_2)+x_1\\
  -m^{-1}a_2(x_1)+(h_0-h_1+1)x_2\\
  \end{pmatrix}
  \end{aligned}
  \end{equation}
  \begin{equation}
  \begin{aligned}
  =\tau_3\circ\sigma_{M_1}\circ\tau_4
  \begin{pmatrix}
  x_1\\
  x_2\\
  \end{pmatrix}.
  \end{aligned}
  \end{equation}
  So $\sigma+e=\tau_3\circ\sigma_{M_1}\circ\tau_4$ is a PP over $\mathbb{F}_{q^2}$ and $\sigma$ is a  CPP over $\mathbb{F}_{q^2}$. Moreover, if $a_1\circ a_2=e,$  then  $\tau_1\circ\tau_2=e.$ Combining (2),  we have that  $\sigma$ is an $r$-regular CPP over $\mathbb{F}_{q^2}.$
  \item It is easy to check that $P_{M}(t)=h(t),$ by (1), $\sigma$ is a  PP over $\mathbb{F}_{q^2}.$  Next, we  prove that $\sigma+e$ is PP over $\mathbb{F}_{q^2}$. Let
  \begin{equation}
  \begin{aligned}
  \tau_5\begin{pmatrix}
  x_1\\
  x_2\\
  \end{pmatrix}
  =
  \begin{pmatrix}
  a_1(x_1)+a_2^{-1}(x_2)\\
  -m^{-1}x_2\\
  \end{pmatrix}
  ,
  \tau_6\begin{pmatrix}
  x_1\\
  x_2\\
  \end{pmatrix}
  =
  \begin{pmatrix}
  a_2(x_1)+mx_2\\
  -mx_2\\
  \end{pmatrix}
  \end{aligned}
  \end{equation}
  and
  \begin{equation*}
    M_1=\begin{pmatrix}
           -h_1+1 & -h_0\\
            1 &  1
            \end{pmatrix}.
  \end{equation*}
  It is easy that $\tau_5,\tau_6,\sigma_{M_1}$ are PPs over $\mathbb{F}_{q^2}$. By calculating, we have
  \begin{equation}
  \begin{aligned}
  (\sigma+e)\begin{pmatrix}
  x_1\\
  x_2\\
  \end{pmatrix}
  =
  \begin{pmatrix}
  a_1((-h_1+1)a_2(x_1)+m(h_0-h_1+1)x_2)+x_1\\
  -m^{-1}a_2(x_1)\\
  \end{pmatrix}
  \end{aligned}
  \end{equation}
  \begin{equation}
  \begin{aligned}
  =\tau_5\circ\sigma_{M_1}\circ\tau_6
  \begin{pmatrix}
  x_1\\
  x_2\\
  \end{pmatrix}.
  \end{aligned}
  \end{equation}
  So $\sigma+e=\tau_5\circ\sigma_{M_1}\circ\tau_6$ is a PP over $\mathbb{F}_{q^2}$ and $\sigma$ is a  CPP over $\mathbb{F}_{q^2}$. Moreover, if $a_1\circ a_2=e,$  then  $\tau_1\circ\tau_2=e.$ Combining (2),  we have that  $\sigma$ is an $r$-regular CPP over $\mathbb{F}_{q^2}.$
\end{enumerate}
\qed\end{proof}

\begin{Remark}
\begin{enumerate}
  \item If $M'=P^{-1}MP$, then $P_{M'}(t)=P_M(t)$. So when $\tau$ is additive, there are many matrices $M$ can be used to construct $r$-regular CPP over $\mathbb{F}_{q^2}.$
  \item In the proof of $(3),$ when $\tau_1\circ\tau_2=e$ and $\tau_1$ is additive,  we have $\sigma+e=\tau_1\circ(\sigma_M+e)\circ\tau_1^{-1},$ and then $\sigma+e$ is a PP over $\mathbb{F}_{q^2}$. But when $\tau_1$ is not additive, $\sigma+e\neq\tau_1\circ(\sigma_M+e)\circ\tau_1^{-1}$ in general, and it is difficult to know when $\sigma+e$ is a PP over $\mathbb{F}_{q^2}$.
  \item In $(4)$ and $(5),$ $a_1$ is any PP over $\mathbb{F}_q$ and $m$ is any element in $\mathbb{F}_q^*$. So we also can give many $r$-regular CPPs over $\mathbb{F}_{q^2}$.
  \item In $(4)$, if $a_1=a_2^{-1}$ is a additive, then $\tau_3\circ\tau_4=e$ and $\sigma+e=\tau_3\circ\sigma_{M_1}\circ\tau_3^{-1}$. Since $\sigma+e=\tau_3\circ\sigma_{M_1}\circ\tau_3^{-1}$ and $\sigma_{M_1}$ have the same cycle structures, we can get some information about the cycle structure of $\sigma+e$. But if $a_1=a_2^{-1}$  is not additive, then $\tau_3\circ\tau_4\neq e$ in general and it is not easy to get the cycle structures of $\sigma+e$.
      \item Similarly, in $(5)$, if $a_1=a_2^{-1}$ is a additive, then $\tau_5\circ\tau_6=e$ and $\sigma+e=\tau_5\circ\sigma_{M_1}\circ\tau_5^{-1}$. Since $\sigma+e=\tau_5\circ\sigma_{M_1}\circ\tau_5^{-1}$ and $\sigma_{M_1}$ have the same cycle structures, we can get some information about the cycle structure of $\sigma+e$. But if $a_1=a_2^{-1}$  is not additive, then $\tau_5\circ\tau_6\neq e$ in general and it is not easy to get the cycle structures of $\sigma+e$.
\end{enumerate}
\end{Remark}

Next, we discuss the univariate form and multivariable form of the same polynomial from the quadratic extension field  $\mathbb{F}_{q^2}$ to itself through the dual basis. For  other relations between univariate forms and multivariable forms, one can see \cite{MP2014,XZZ2022}.

Let $\operatorname{Tr}_q^{q^2}(x)=x+x^{q}$ be the  \emph{trace function} from $\mathbb{F}_{q^2}$ to $\mathbb{F}_q.$
Let $\alpha\in\mathbb{F}_{q^2}\backslash \mathbb{F}_{q}$.
Then  $\mathbb{F}_{q^2}$ is a vector space with dimension $2$ over $\mathbb{F}_q$ and $\{\alpha_1=1, \alpha_2= \alpha\}$ is a basis. Each element $\mathbb{F}_{q^2}$ can be uniquely represented as
\begin{equation*}\label{isomorphism}
x=x_1+\alpha x_2=\alpha_1x_1+\alpha_2 x_2
\end{equation*}
where $x_1, x_2\in \mathbb{F}_q.$
Let  $\{\beta_1, \beta_2\}$ be the dual basis of $\{\alpha_1, \alpha_2\}$, see \cite[p.58]{R2003}. Then  for  $1 \leq i, j \leq 2,$  we have
\begin{equation*}
\operatorname{Tr}_q^{q^2}\left(\alpha_{i} \beta_{j}\right)=\left\{\begin{array}{ll}
1 & ,\text { for } i=j, \\
0 & ,\text { for } i \neq j.
\end{array}\right.
\end{equation*}
By calculating, we have
\begin{equation*}
  \operatorname{Tr}_q^{q^2}\left(x \beta_{1}\right)=\operatorname{Tr}_q^{q^2}\left(\alpha_{1} \beta_{1}\right)x_1+\operatorname{Tr}_q^{q^2}\left(\alpha_{2} \beta_{1}\right)x_2=x_1.
\end{equation*}
and
\begin{equation*}
  \operatorname{Tr}_q^{q^2}\left(x \beta_{2}\right)=\operatorname{Tr}_q^{q^2}\left(\alpha_{1} \beta_{2}\right)x_1+\operatorname{Tr}_q^{q^2}\left(\alpha_{2} \beta_{2}\right)x_2=x_2.
\end{equation*}
That is to say:
\begin{equation*}
\left\{\begin{array}{ll}
x_1=\operatorname{Tr}_q^{q^2}\left(x \beta_{1}\right)=x{\beta_1}+x^q{\beta_1}^q, \\
x_2=\operatorname{Tr}_q^{q^2}\left(x \beta_{2}\right)=x{\beta_2}+x^q{\beta_2}^q.
\end{array}\right.
\end{equation*}

Now, we can rewrite the ($r$-regular) CPPs in Theorem \ref{main theorem} (4) or (5) in the univariate form. Let $\tau_1(x_1,x_2)=(a_1(x_1),x_2)$ and $\tau_2(x_1,x_2)=(a_2(x_1),x_2)$, where $a_1, a_2$ are any PPs over $\mathbb{F}_q$. Let $M=\big(\begin{smallmatrix}m_1 & m_2\\ m_3 & m_4\end{smallmatrix}\big)\in M_{2\times2}(\mathbb{F}_q)$.  By calculating, we have
  \begin{equation}
  \begin{aligned}
  \sigma\begin{pmatrix}
  x_1\\
  x_2\\
  \end{pmatrix}
  =
  \begin{pmatrix}
  a_1(m_1a_2(x_1)+m_2x_2)\\
  m_3a_2(x_1)+m_4x_2\\
  \end{pmatrix}
  \end{aligned}.
  \end{equation}
Through the basis $\{\alpha_1=1, \alpha_2= \alpha\}$,  we can get  a ($r$-regular) univariate complete permutation polynomial $f(x)=f_{\sigma}(x)\in \mathbb{F}_{q^2}[x]$ such that
\begin{eqnarray*}
   f(x) &=& a_1(m_1a_2(x_1)+m_2x_2)+\alpha(m_3a_2(x_1)+m_4x_2) \\
        &=& a_1(m_1a_2(x{\beta_1}+x^q{\beta_1}^q)+m_2(x{\beta_2}+x^q{\beta_2}^q)) \\
        & & +\alpha(m_3a_2(x{\beta_1}+x^q{\beta_1}^q)+m_4(x{\beta_2}+x^q{\beta_2}^q)).
\end{eqnarray*}

\section {\bf{Concluding remarks}}\label{section Concluding remarks}
In this paper, we  give two constructions of $r$-regular  CPPs over quadratic extension fields when $q^2\equiv 1 \pmod r$. These $r$-regular  CPPs have  the forms $\sigma=\tau\circ\sigma_M\circ\tau^{-1}$ where $\tau$ is a PP over a quadratic extension field $\mathbb{F}_{q^2}$ and $\sigma_M$ is an invertible linear map  over  $\mathbb{F}_{q^2}$. In further research, we want to give $r$-regular  CPPs over other extension fields.

 {}

\begin{thebibliography}{}

 \bibitem{A1969} Ahmad S.: Cycle structure of automorphisms of finite cyclic groups. J. Comb. Theory 6, 370-374 (1969).
 \bibitem{B2003} Biryukov A.: Analysis of involutional ciphers: Khazad and Anubis. Fast Softw. Encryption 2887, 45-53 (2003).
\bibitem{CR2015} Canteaut A., Roue J.: On the behaviors of affine equivalent S-boxes regarding differential and linear attacks, in: Advances in Cryptology - EUROCRYPT 2015 - 34th Annual International Conference on the Theory and Applications of Cryptographic Techniques, Sofia, Bulgaria, April 26-30, 2015, in: Lecture Notes in Computer Science, Part I, vol. 9056, Springer, pp. 45-74 (2015).
\bibitem{CMS2016} Charpin P., Mesnager S., Sarkar S.: Involutions over the Galois field $\mathbb{F}_{2^n}$. IEEE Trans. Inf. Theory 62 (4), 2266-2276  (2016).
\bibitem{CWZ2021}Chen Y., Wang L., Zhu S.: On the constructions of n-cycle permutations. Finite Fields Appl. 73, 101847 (2021).
\bibitem{CM2018} Coulter R.S., Mesnager S.: Bent functions from involutions over $\mathbb{F}_{2^n}$. IEEE Trans. Inf. Theory 64 (4), 2979-2986 (2018).
\bibitem{DL2008} Diffie W., Ledin G. (translators): SMS4 encryption algorithm for wireless networks. https://eprint.iacr.org/2008/329.pdf.
\bibitem{FFZ2011} Feng D., Feng X., Zhang W., et al.: Loiss: a byte-oriented stream cipher. In: IWCC'11 Proceedings of the Third International Conference on Coding and Cryptology,   109-125. Springer, New York (2011).
\bibitem{F1982} Fredricksen H.: A survey of full length nonlinear shift register cycle algorithms. SIAM Rev. 24(2), 195-221 (1982).
\bibitem{G1962} Gallager R.: Low-density parity-check codes. IRE Trans. Inf. Theory 8 (1),  21-28  (1962).
\bibitem{G1967} Golomb S.W.: Shift Register Sequences. Holden-Day Inc, Laguna Hills (1967).
\bibitem{GG2005} Golomb S.W., GongG.:Signal Design for Good Correlation. For Wireless Communication, Cryptography, and Radar. Cambridge University Press, New York (2005).
\bibitem{L2002} Lang, S.: Algebra. Springer New York, (2002).
\bibitem{LM1991} Lidl R., Mullen G.L.: Cycle structure of Dickson permutation polynomials. Math. J. Okayama Univ. 33, 1-11 (1991).
\bibitem{LWWC2022}Lu W., Wu X., Wang Y., Cao X.: A general construction of regular complete permutation polynomials, arXiv preprint, arXiv:2212.12869, 2022.
\bibitem{M1942}Mann H.B.: The construction of orthogonal Latin squares. Ann. Math. Stat. 13(4), 418-423 (1942).
\bibitem{MM2009}Markovski S., Mileva A.: Generating huge quasigroups from small non-linear bijections via extended Feistel function. Quasigroups Relat. Syst. 17(1), 91-106 (2009).
\bibitem{Ma1997}Matsui M.: New block encryption algorithm MISTY.  In: Fast Software Encryption-FSE'97. Lect. Notes Comput. Sci, vol. 1267, pp. 54-68. Springer, New York (1997).
\bibitem{M2016} Mesnager S.: On constructions of bent functions from involutions, in: 2016 IEEE International Symposium on Information Theory (ISIT), IEEE,  110-114 (2016).
\bibitem{MM2012a}Mileva A., Markovski S.: Quasigroup representation of some Feistel and generalized Feistel ciphers. In: ICT Innovations 2012. Advances in Intelligent Systems and Computing, vol. 207,  161-171. Springer, Berlin (2012).
\bibitem{MM2012b}Mileva A., Markovski S.: Shapeless quasigroups derived by Feistel orthomorphisms. Glas. Mat. 47(67), 333-349 (2012).
\bibitem{M1995}Mittenthal L.: Block substitutions using orthomorphic mappings. Adv. Appl. Math. 16(10), 59-71 (1995).
\bibitem{M1997}Mittenthal L.: Nonlinear dynamic substitution devices and methods for block substitutions employing coset decompositions and direct geometric generation. US Patent 5647001 (1997).
\bibitem{MP2014}Muratovic-Ribic A., Pasalic E.: A note on complete polynomials over finite fields and their applications in cryptography. Finite Fields Appl. 25, 306-315 (2014).
\bibitem{M2021}Muratovic-Ribic, A., On generalized strong complete mappings and mutually orthogonal Latin squares. Ars Mathematica Contemporanea 21(2) (2021).
\bibitem{NIST1993}National Institute of Standards and Technology.: Data Encryption Standard, FIPS Publication 46-2 (1993).
\bibitem{NR1982}Niederreiter H., Robinson K.H.: Complete mappings of finite fields. J. Aust. Math. Soc. A 33(2), 197-212 (1982).
\bibitem{RC2004}Rubio I., Corrada C.: Cyclic decomposition of permutations of finite fields obtained using monomials,. Finite Fields and Applications, LNCS 2948,   254-261, Springer, New York (2004).
\bibitem{RMCC2008}Rubio I., Mullen G.L., Corrada C., Castro F.N.: Dickson permutation polynomials that decompose in cycles of the same length. Contemp. Math. 461, 229-240 (2008).
\bibitem{R2003}Rudolf Lidl, Harald Niederreiter: Finite fields. Encyclopedia of Mathematics and ITS Applications, (2003).
\bibitem{SSP2012}Sakzad A., Sadeghi M.R., Panario D.: Cycle structure of permutation functions over finite fields and their applications. Adv. Math. Commun. 6(3), 347-361 (2012).
\bibitem{SV1995}Schnorr C.P., Vaudenay S.: Black box cryptanalysis of hash networks based on multipermutations. In: Advances in Cryptology-Eurocrypt'94,  47-57. Springer, New York (1995).
\bibitem{SG2012}Stanica P., Gangopadhyay S., Chaturvedi A., Gangopadhyay A.K., Maitra S.: Investigations on bent and negabent functions via the negaHadamard transform. IEEE Trans. Inf. Theory 58, 4064-4072 (2012).
\bibitem{V1994}Vaudenay S.: On the need for multipermutations: cryptanalysis of MD4 and SAFER. In: Fast Software Encryption-FSE'94. Lect. Notes Comput. Sci., vol. 1008,  286-297. Springer, New York (1994).
\bibitem{V1999}Vaudenay S.: On the Lai-Massey scheme. In: Advances in Cryptology-ASIACRYPT-99. Lect. Notes Comput. Sci., vol. 1716, 8-19. Springer, New York (1999).
\bibitem{XLZH2018}Xu, X., Li, C., Zeng, X., Helleseth, T.:  Constructions of complete permutation polynomials. Des. Codes Cryptogr. 86, 2869-2892 (2018).
\bibitem{XZZ2022}Xu, X., Zeng, X., Zhang, S: Regular complete permutation polynomials over $\mathbb{F}_{2^n}$. Des. Codes Cryptogr. 90, 545-575 (2022).
\bibitem{ZHC2015}Zha Z., Hu  L., Cao X.: Constructing permutations and complete permutations over finite fields via subfield-valued polynomials. Finite Fields Appl. 31  162-177  (2015).













\end{thebibliography}
\end{document}